\documentclass[fleqn,11pt,oneside]{article}

\usepackage[ruled,vlined,linesnumbered]{algorithm2e}
\usepackage{graphicx,hyperref}
\usepackage[export]{adjustbox}
\usepackage{bbm}
\usepackage{bm}
\usepackage{amsmath,amsthm,amssymb,amscd}
\usepackage{accents}
\usepackage{subfig}
\usepackage{multirow}
\usepackage{pdfpages}
\usepackage{setspace} 
\usepackage{caption}
\usepackage[ top=1.1in, 
             bottom=1.1in, 
             left=1.6in, 
             right=1.1in, 
             driver=dvipdfm, 
             ]{geometry}
\usepackage{microtype}

\usepackage{tikz}
\usetikzlibrary{math,calc}
\usepackage{pgfplots}

\usepackage{tabu}
\newtabulinestyle{mydashed=on 2.8pt off 2.8pt}

\def\1{\mathbbm{1}}
\def\bb{\begin{equation}}
\def\ee{\end{equation}}

\newtheorem{theorem}{Theorem}[section]
\newtheorem{lemma}[theorem]{Lemma}

\newtheorem{observe}{Observation}[section]
\newtheorem{remark1}[observe]{Remark}

\newtheorem{formula}{Formula}[section]

\newenvironment{remark}{\begin{remark1} \rm}{\end{remark1}}

\def\qed{\hfill$\blacksquare$\\} \renewenvironment{proof}{\noindent {\bf 
Proof.}}{\qed}

\def\R{\mathbbm{R}}
\def\N{\mathbbm{N}}
\def\E{\mathbbm{E}}
\def\var{\mathrm{var}}
\def\cov{\mathrm{cov}}

\title{A Fast Linear Regression via SVD and Marginalization}
\author{Philip Greengard\thanks{Columbia University}, 
Andrew Gelman\thanks{Department of Statistics and Department of Political Science, 
Columbia University},
Aki Vehtari\thanks{Department of Computer Science, Aalto University}}
\date{November 9, 2020}
\begin{document}

\maketitle

\begin{abstract}
We describe a numerical scheme for evaluating the posterior 
moments of Bayesian linear regression models with partial 
pooling of the coefficients. The principal analytical tool of the evaluation is a change of 
basis from coefficient space to the space of singular vectors 
of the matrix of predictors. After this change of basis
and an analytical integration, we reduce the problem of 
finding moments of a density over $k + m$ dimensions, to 
finding moments of an $m$-dimensional density, where $k$ is the
number of coefficients and $k+m$ is the dimension of the posterior.
Moments can then be computed using, for example, MCMC, the trapezoid 
rule, or adaptive Gaussian quadrature. 
An evaluation of the SVD of the matrix of predictors is the 
dominant computational cost and is performed once during the 
precomputation stage. We demonstrate numerical results of the algorithm. 
The scheme described in this paper generalizes naturally 
to multilevel and multi-group hierarchical regression models
where normal-normal parameters appear.
\end{abstract}

\section{Introduction}

Linear regression is a ubiquitous tool for statistical 
modeling in a range of applications including social
sciences, epidemiology, biochemistry, and environmental sciences 
(\cite{bda, 5, 15, 25, 35}). 

A common bottleneck for applied statistical modeling workflow is the 
computational cost of model evaluation. 
Since posterior distributions in statistical models 
are often high dimensional and computationally 
intractable, various techniques have been used to approximate 
posterior moments. 
Standard approaches often involve a variety of techniques including
Markov chain Monte Carlo (MCMC) or using a suitable approximation
of the posterior.

In this paper, we describe an approach for reducing the computational costs 
for a particular
class of regression models --- those that contain parameters $\theta \in \R^k$
such that $\theta$ has a normal prior and normal likelihood. 
These models represent only a subset of regression models that 
appear in applications. 
We focus our attention in this paper on normal-normal 
models because they have well known 
analytical properties and are more computationally tractable
than the vast majority of multilevel models. 
A broader class of models, including logistic regression,
contain distributions that are less amenable to 
the techniques of this paper and will require other
analytical and computational tools. 
Mathematically, marginalization of normal-normal parameters is 
well-known and has been applied to the posterior by, for example, 
\cite{lindley}. Our contribution is to provide a stable, accurate,
and fast algorithm for marginalization. 

The primary numerical tool used in the 
algorithm is the singular value decomposition (SVD) of the data
matrix. As a mathematical and statistical tool, SVD has been known since 
at least 1936 (see \cite{eckart}). Use of the
SVD as a practical and efficient numerical algorithm only started gaining
popularity much later, with the first widely used scheme
introduced in \cite{golub}. Due in large part to advances in 
computing power, use of the SVD as a tool in applied 
mathematics, statistics, and data science has been gaining significant
popularity in recent years, however efficient
evaluation of SVDs and related matrix decompositions is still an active 
area of research (see \cite{hastie}, \cite{halko}, \cite{shamir}). 

Similar schemes to ours are used in the 
software packages lme4 (\cite{lme4}) and INLA (\cite{inla}).
There are several differences between the problems they address and
their computational techniques, and those that we shall discuss here. 
While lme4 finds maximum likelihood and restricted maximum likelihood
estimates, our goal is to find posterior moments. 
The software package INLA uses Laplace approximation on the posterior 
for a general choice of likelihood functions, whereas 
our algorithm is focused on fast and accurate solutions for only 
a particular class of densities: those with normal-normal parameters.

The approach presented in this paper analytically marginalizes the normal-normal 
parameters of a model using a change of variables.
After marginalization, posterior moments can be computed using 
standard techniques on the lower-dimensional density. 
In particular, for a model that contains $k+m$ total variables, $k$
of which are normal-normal, our scheme converts the problem 
of evaluating expectations of a density in 
$k+m$ dimensions to finding expectations of an $m$-dimensional density.
After marginalization, we evaluate the $m$-dimensional 
posterior density in $O(k)$ operations. Without the change of 
variables, standard evaluation of marginal densities that relies on 
determinant evaluation requires at least $O(k^3)$ operations. 

We illustrate our scheme on the problem of evaluating the 
marginal expectations of the unnormalized density
\bb\label{10}
\hspace*{-3em}
q(\sigma_1, \sigma_2, \beta) = 
\sigma_1^{-(k+1)} \sigma_2^{-n}\exp\left(
- \gamma\log^2(\sigma_1) - \frac{\sigma_2^2}{2} 
- \frac{\|X\beta - y\|^2}{2\sigma_2^2}
- \frac{\|\beta \|^2}{2\sigma_1^2} \right),
\ee
where $\gamma>0$ is a constant, $\sigma_1, \sigma_2 > 0$, and $\beta \in \R^k$. 
We assume that $X$ is a fixed $n \times k$ matrix, $y \in \R^n$ is fixed, 
and the normalizing constant of (\ref{10})
is unknown. For fixed $n,k \in \N$, the algorithm is nearly identical 
when $X$ is an $n \times k$ matrix to when $X$ is a $k \times n$ 
matrix. In the case where $k \gg n$, see \cite{kwon} for a similar 
approach. 
Using the standard notation of Bayesian models, density $q$ is the 
unnormalized posterior of the model
\bb\label{16}
\begin{split}
&\sigma_1 \sim \text{lognormal}(0, \sqrt{\gamma})\\
&\sigma_2 \sim \text{normal}^+(0, 1)\\
&\beta \sim \text{normal}(0, \sigma_1)\\
&y \sim \text{normal}(X\beta, \sigma_2).
\end{split}
\ee
In Appendix A, we include Stan code that can be used to 
sample from density (\ref{10}). 

Statistical model (\ref{16}) is a standard model of Bayesian statistics
and appears when seeking to model an outcome, $y$, as a linear combination 
of related predictors, the columns of $X$. In \cite{5}, these models 
are described in detail and are used in the estimation of the 
distribution of radon levels in houses in Minnesota.

Density (\ref{10}) is also closely related to posterior densities that 
appear in genome-wide association studies (GWAS; see \cite{zhu},
\cite{1101}, \cite{azevedo}) which can be used to identify
genomic regions containing genes linked with a specific trait, such 
as height. Using the 
notation of (\ref{10}), each row of matrix $X$ corresponds to a person, each 
column of $X$ represents a genomic location, entries of $X$ 
indicate genotypes, and $y$ corresponds to the trait. Due to technical 
advances in genome sequencing over the last ten years, it is now
feasible to collect large amounts of sequencing data. 
GWAS models can contain data on up to millions of people 
and often between hundreds and thousands of genome locations 
(see \cite{karlsson}). 
As a result, efficient computational tools are required for model evaluation. 

The number of operations required by the scheme of this paper 
scales like $O(nk^2)$ with a small constant. 
The key analytical tool is a change of 
variables of $\beta$ such that the terms, 
\bb
-\frac{1}{2\sigma_2^2}\|X\beta - y\|^2 - \frac{1}{2\sigma_1^2}\| \beta \|^2,
\ee
in (\ref{10}) are converted to a diagonal quadratic form in $\R^k$. 
After that change of variables, expectations over $q$ are 
analytically converted from integrals over $\R^{k+2}$ to 
integrals over $\R^2$. The remaining $2$-dimensional integrals 
can be computed to high accuracy using classical numerical 
techniques including, for example, adaptive Gaussian quadrature 
or even the $2$-dimensional trapezoid rule.

The schemes used to evaluate the expectations of 
(\ref{10}) generalize naturally to evaluation of expectations of 
multilevel and multigroup posterior distributions including, for
example, the two-group posterior of the form,
\bb\label{1220}
\hspace{-3em}
q(\sigma_1,\sigma_2,\sigma_3, \mu, \beta) 
= \exp\left( -\frac{1}{2\sigma_1^2} 
\|X\beta - y\|^2 - \frac{1}{2\sigma_2^2}\sum_{i=1}^{k_1} (\mu - \beta_i)^2 
- \frac{1}{2\sigma_3^2}\sum_{i=k_1 + 1}^{k_1 + k_2} \beta_i^2 
\right),
\ee
where $X$ is a $n \times k$
matrix, $y \in \R^n$,
$k_1$ and $k_2$ are non-negative integers satisfying $k_1 + k_2 = k$, and
the vector $t \in \R^m$.

For models where $m$ is large, MCMC can be used to evaluate
the $m$-dimensional expectations, with, for example, Stan \cite{stan}.
The $m$-dimensional distribution has two qualities that make it 
preferable to its high-dimensional counterpart. First, it requires 
only $O(k)$ operations to evaluate the integrand, and second, the 
geometry of the $m$ dimensional marginal distribution will allow for 
more efficient sampling. 

The structure of this paper is as follows. In the following section 
we describe the analytic integration that transforms (\ref{10}) from 
a $k+2$-dimensional problem to a $2$-dimensional problem. 
Section \ref{s30} includes formulas that will allow for the 
evaluation of posterior moments using the $2$-dimensional density.
In Sections \ref{s45} and \ref{s47} we provide formulas 
for evaluating covariances of (\ref{10}).
In Section \ref{s50}, we discuss the numerical results of 
the implementation of the algorithm. 
Conclusions and generalizations of the algorithm of  
this paper are presented in Section \ref{s55}. 
Appendix A provides Stan code that can be used to sample
from (\ref{10}), and Appendix B includes proofs of the 
formulas of this paper. 
\section{Analytic Integration of \texorpdfstring{$\beta$}{b}}\label{s20}
In this section, we describe how we analytically marginalize the 
normal-normal parameter $\beta$ of density (\ref{10}). 
We include proofs of all formulas in Appendix B. 

We start by in marginalizing $\beta$ using a
change of variables that converts
the quadratic forms in (\ref{10}) into diagonal quadratic forms. 
The resulting integral in the new
variable, $z$, is Gaussian, and the coefficients of $z_i$
and $z_i^2$ are available analytically. The change of variables 
is given by the right orthogonal 
matrix of the singular value decomposition (SVD) of $X$. That is,
we set 
\bb\label{735}
z = V^t\beta
\ee
where the SVD of $X$ is
\bb\label{745}
X = UDV^t.
\ee
We define $\lambda_i$ to be the $i^\text{th}$ element of the
diagonal of $D$. The elements of diagonal need not be sorted. 
After this change of variables, we obtain the following 
identity for the last two terms of (\ref{10}). A proof can
be found in Lemma \ref{170} in Appendix B. 
\begin{formula}
\bb\label{330_0}
\hspace{-3em}
-\frac{1}{2\sigma_2^2}\| X\beta - y\|^2 - \frac{1}{2\sigma_1^2} \|\beta\|^2 = 
a_0 + \sum_{i=1}^k 
a_{2,i}\left(z_i - \frac{a_{1,i}}{2a_{2,i}}\right)^2 
+ \frac{a_{1,i}^2}{4a_{2,i}}
\ee
where
\bb\label{340_0}
a_{2,i} = \frac{\lambda_i^2}{2\sigma_2^2} + \frac{1}{2\sigma_1^2},
\ee
\bb\label{350_0}
a_{1,i} = \frac{w_i}{\sigma_2^2},
\ee
and
\bb\label{360_0}
a_{0} = -\frac{y^ty}{2\sigma_2^2}
\ee
where
\bb\label{110_0}
w = V^tX^ty.
\ee
\end{formula}
After performing the change of variables $z=V^t\beta$ and using 
(\ref{330_0}), we now have an 
expression for density (\ref{10}) in a form that allows us to 
use the well-known properties of a Gaussian with diagonal covariance. 
The following identity uses these properties and provides 
a formula for analytically reducing expectations of (\ref{10}) 
from integrals over $k+2$ dimensions
to integrals over $2$ dimensions. After the formula is applied,
we have a new density, $\tilde{q}$, over only $2$ dimensions. 
See Theorem \ref{290} in Appendix B for a proof. 
\begin{formula}\label{290_0}
For all $\sigma_1,\sigma_2>0$ we have
\bb\label{310_0}
\int_{\R^k} q(\sigma_1, \sigma_2, \beta) d\beta = 
\tilde{q}(\sigma_1, \sigma_2)
\ee
where $\tilde{q}(\sigma_1, \sigma_2)$ is defined by the formula
\bb\label{1270_0}
\hspace{-3em}
\tilde{q}(\sigma_1, \sigma_2) = 
\sigma_1^{-(k+1)} \sigma_2^{-n}
\exp\left( 
-\gamma\log^2(\sigma_1) - \frac{\sigma_2^2}{2} 
+a_0 + \sum_{i=1}^k \frac{a_{1,i}^2}{4a_{2,i}}
\right)
\prod_{i=1}^k\frac{1}{\sqrt{2a_{2,i}}}
\ee
where $a_{2,i}$ is defined in (\ref{340_0}), $a_{1,i}$ is defined in 
(\ref{350_0}), $a_0$ is defined in (\ref{360_0}), and $\gamma$ is a 
constant. 
\end{formula}
\begin{remark}
Certain Bayesian models might contain correlated priors on 
$\beta$ that will result in posteriors such as (\ref{1215}) of  Section~\ref{s45}. 
For such models, we perform the change of variables that 
uses the fact that two diagonal forms
over $\beta$ can be simultaneously diagonalized.
\end{remark}

We include in Figure \ref{715} a plot
of the density of $q$ as a function of $\sigma_1$ and 
$\beta_1$ for fixed $\sigma_2$ and randomly chosen $X$ and $y$.
Figure \ref{710} shows a plot of $q$ as a function of 
$\sigma_2$ and $\beta$ for fixed $\sigma_1$.
Figure \ref{716} provides an illustration of $\tilde{q}$, obtained 
after the change of variables and marginalization described 
in this section.

%
%
\begin{figure}[!ht]
\centering
\begin{minipage}{.75\textwidth}
  \includegraphics[width=\linewidth]{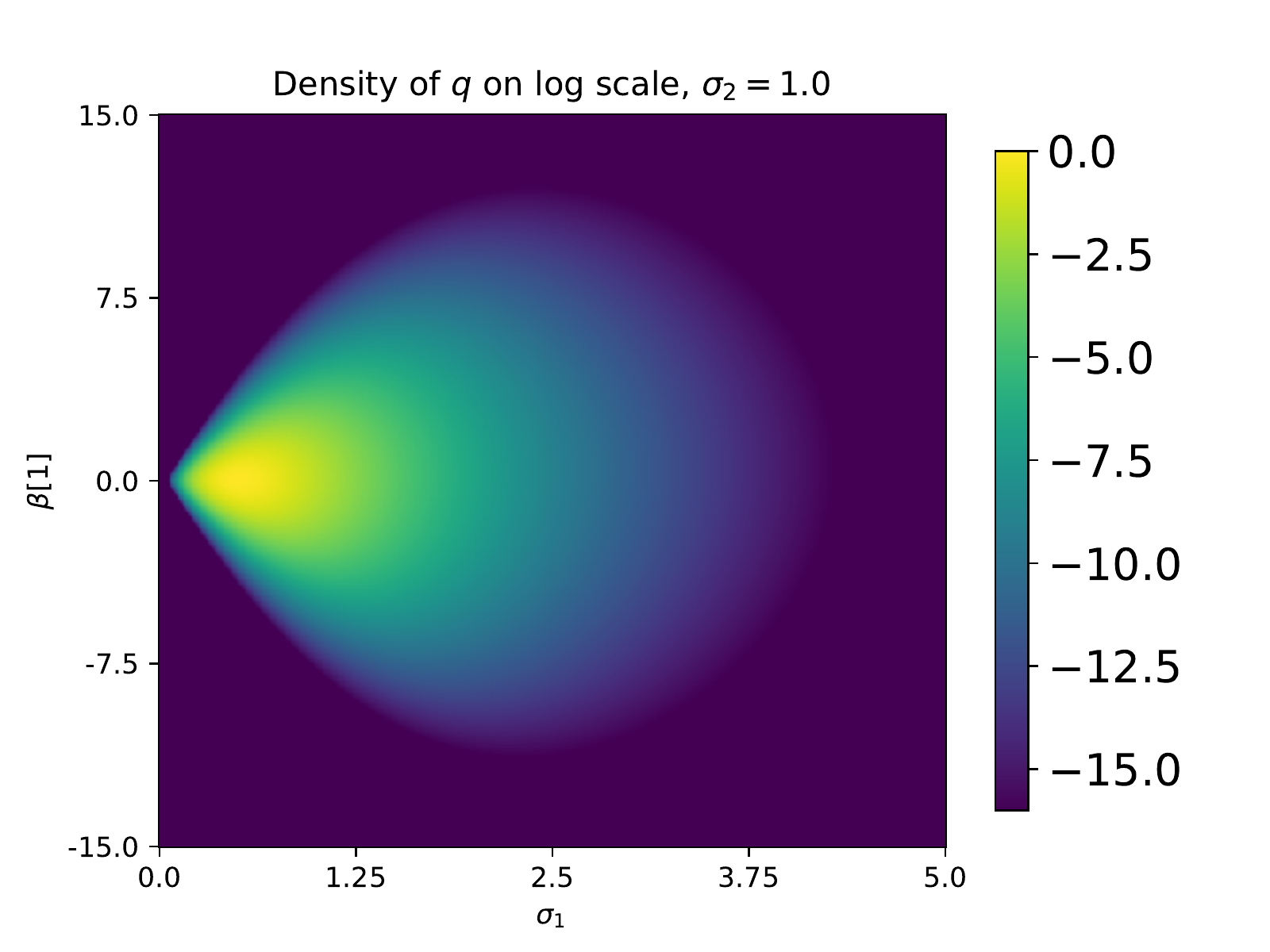}
  \caption{\em Density of $q$ (see (\ref{10})) with respect to
$\sigma_1$ and $\beta_1$, where $\gamma=8$, 
$n=100$, $k=10$, and data were randomly generated.}
  \label{715}
\end{minipage}%

\begin{minipage}{.75\textwidth}
  \includegraphics[width=\linewidth]{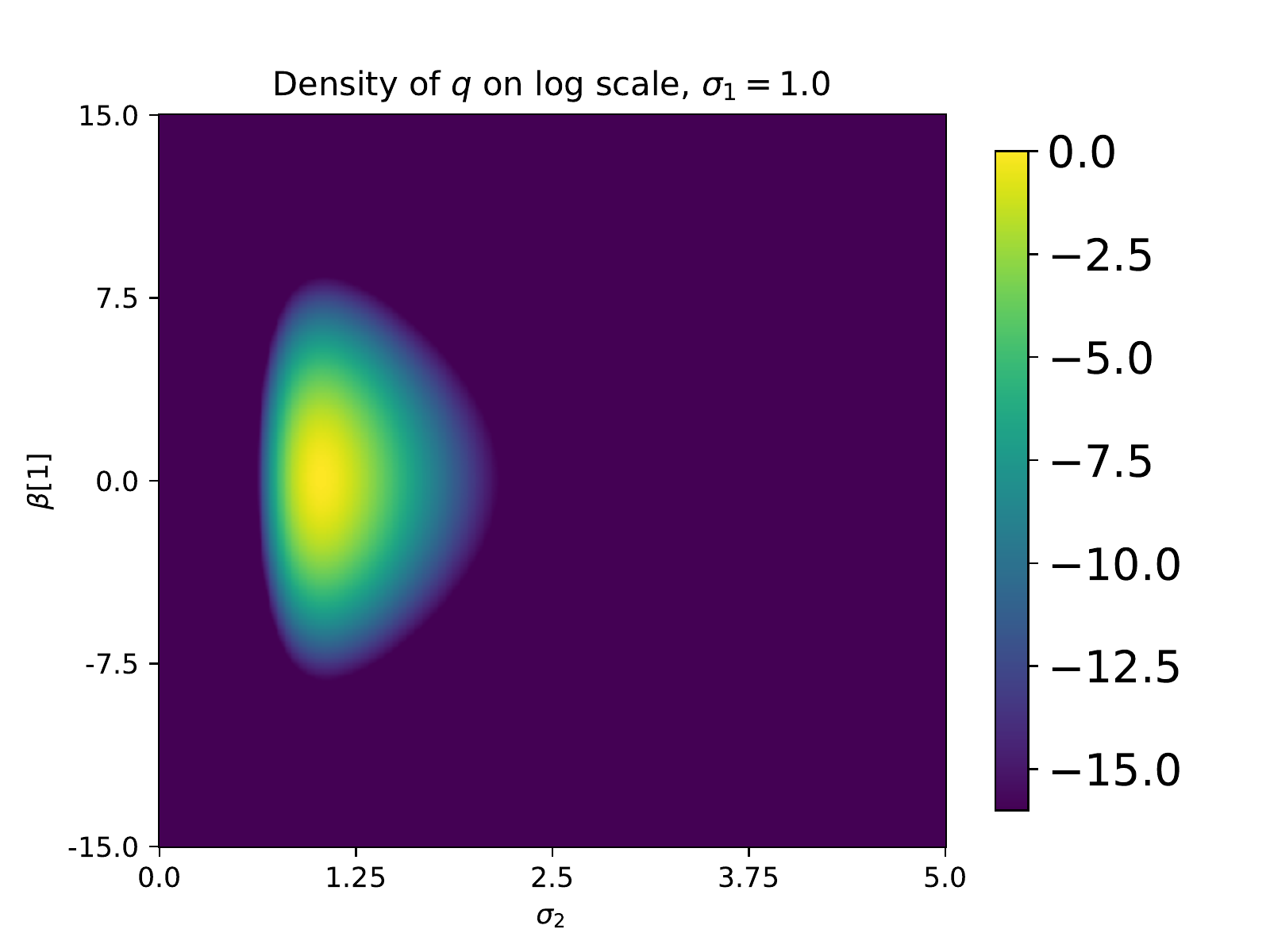}
  \caption{\em Density of $q$ (see (\ref{10})) with respect to 
$\sigma_2$ and $\beta_1$, for the same parameters as 
Figure \ref{715}.}
  \label{710}
\end{minipage}
\end{figure}

%
%
\begin{figure}[!ht]
\centering
  \includegraphics[width=\textwidth]{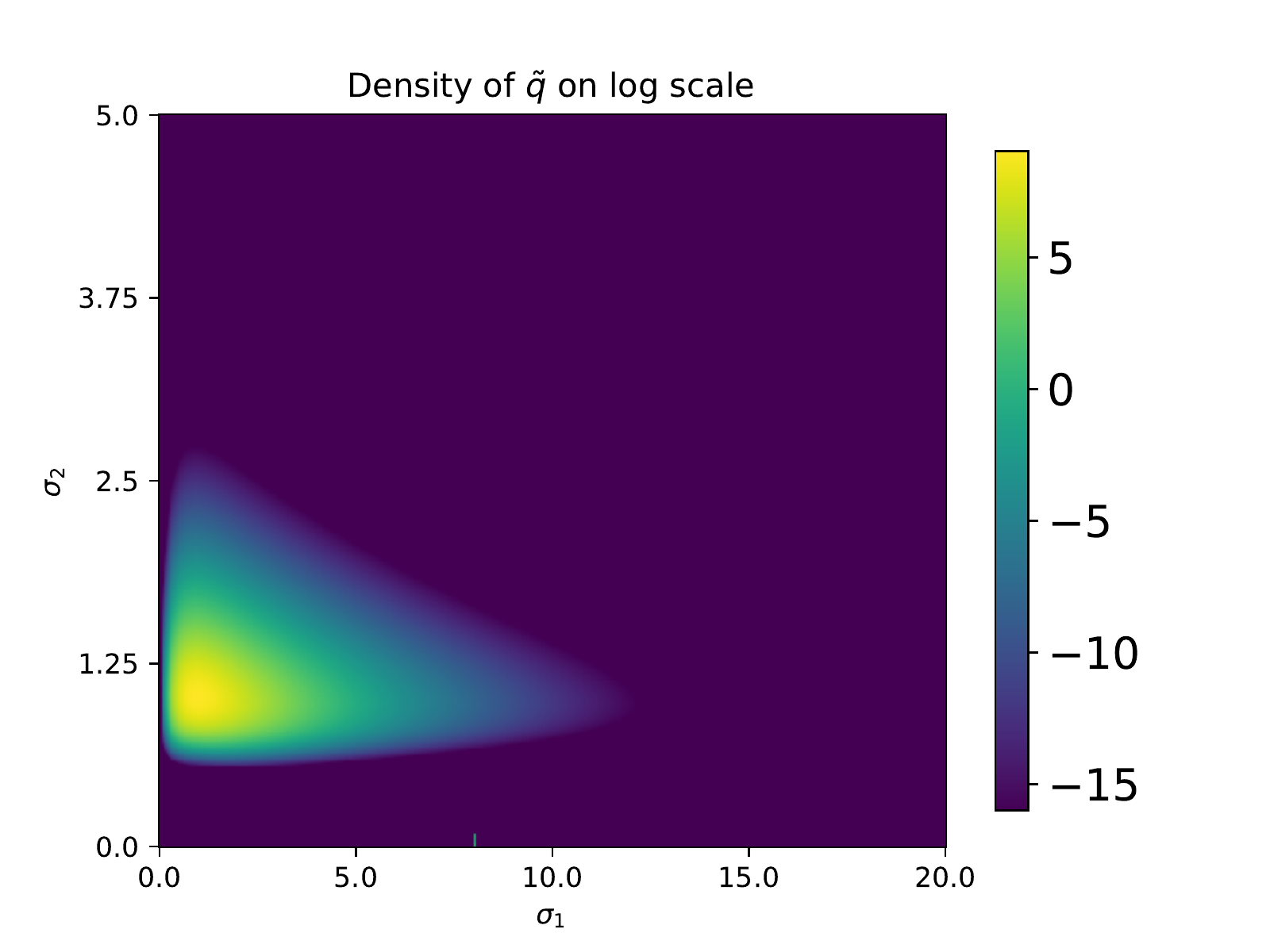}
  \caption{\em Density of $\tilde{q}$ (see (\ref{1270_0})) using 
the same $q$ as Figure \ref{715}, where $n=100$, $k=10$, and 
data were randomly generated.}
  \label{716}
\end{figure}
\section{Evaluation of Posterior Means}\label{s30}
Now that we have reduced the $k+2$-dimensional density $q$
to the $2$-dimensional density $\tilde{q}$, it remains to 
recover the posterior moments of $q$ using $\tilde{q}$. 
We first observe that moments of $\sigma_1$ and $\sigma_2$
with respect to $q$ are equivalent to moments of $\sigma_1$ and 
$\sigma_2$ over $\tilde{q}$. That is,
\bb
\E_{q}(\sigma_1) = \E_{\tilde{q}}(\sigma_1)
\ee
and 
\bb
\E_{q}(\sigma_2) = \E_{\tilde{q}}(\sigma_2).
\ee

As for moments of $\beta$, we use (\ref{1270_0}) and
standard properties of Gaussians to obtain the 
following formula. 
\begin{formula}\label{380_0}
For all $\sigma_1,\sigma_2>0$,
\bb\label{390_0}
\int_{\R^k} z_i q(\sigma_1, \sigma_2, \beta) d\beta = 
\frac{a_{1,i}}{2a_{2,i}}
\tilde{q}(t)
\ee
where $q$ is defined in (\ref{10}), $\tilde{q}$ is defined in 
(\ref{1270_0}),
$a_{2,i}$ is defined in (\ref{340_0}), and $a_{1,i}$ is defined in 
(\ref{350_0}). 
\end{formula}
As an immediate consequence of (\ref{390_0}), we are 
able to evaluate the posterior expectation of $z$
as an expectation of a $2$-dimensional density:
\bb\label{810}
\E_{q}(z_i) = \E_{\tilde{q}}(\frac{a_{1,i}}{2a_{2,i}}).
\ee
We then transform
those expectations back to expectations over the desired basis, 
$\beta$ using the matrix $V$ computed in (\ref{745}). 
Specifically, using linearity of expectation and (\ref{810}), we know
\bb\label{570}
\begin{split}
\E_{q}((\beta_1,\dots,\beta_k)^t) 
&= \E_{q}(VV^t(\beta_1,\dots,\beta_k)^t)\\
&= V\E_{q}(V^t(\beta_1,\dots,\beta_k)^t)\\
&= V\E_{q}((z_1,\dots,z_k)^t)\\
&= V\E_{\tilde{q}}\left(\left(
\frac{a_{1,1}}{2a_{2,1}},\dots,
\frac{a_{1,k}}{2a_{2,k}}
\right)^t\right).
\end{split}
\ee
\section{Covariance of \texorpdfstring{$\beta$}{b}}\label{s45}
In addition to facilitating the rapid evaluation of posterior means,
the change of variables described in Section \ref{s20} is also 
useful for the evaluation of higher moments. 

Equation (\ref{330_0}) shows that after the change of variables
from $\beta$ to $z$, the resulting density is a Gaussian 
in $z$ with a diagonal covariance matrix. 
Additionally, for each $z_i$, using equation (\ref{330_0}) and 
standard properties of Gaussians, we have the following identity.
\begin{formula}\label{385}
For all $\sigma_1, \sigma_2>0$, we have
\bb\label{390_1}
\int_{\R^k} (z_i-\mu_{z_i})^2 q(\sigma_1, \sigma_2, \beta) d\beta = 
(2a_{2,i})^{-1}\tilde{q}(\sigma_1, \sigma_2)
\ee
where $\mu_{z_i}$ is the expectation of $z_i$, $\tilde{q}$ is 
defined in (\ref{1270_0}), and $a_{2,i}$ is defined in (\ref{340_0}). 
\end{formula}
The second moments of the posterior of $\beta$ are obtained as a
linear transformation of the posterior variances of $z$. In particular, 
denoting the expectation of $\beta$ by $\mu_\beta$ and the expectation
of $z$ by $\mu_z$, we have
\bb\label{1010}
\begin{split}
\E((\beta-\mu_\beta)(\beta-\mu_\beta)^t) 
&= VV^t\E((\beta-\mu_\beta)(\beta-\mu_\beta)^t)VV^t \\
&= V\E(V^t(\beta-\mu_\beta)(\beta-\mu_\beta)^tV)V^t \\
&= V\E((z-\mu_z)(z-\mu_z)^t)V^t.
\end{split}
\ee
We observe that due to the independence of all $z_i$, 
\bb
\E((z-\mu_z)(z-\mu_z)^t)
\ee
is diagonal and we can therefore evaluate the 
$k \times k$ posterior covariance matrix of $\beta$ by evaluating 
 $\var(z_i)$ for $i=1,...,k$ and then applying two 
orthogonal matrices. Specifically, 
combining Formula \ref{385} and (\ref{1010}), we obtain 
\bb\label{1020}
\begin{split}
\cov(\beta) 
&= V\E_{\tilde{q}}\left(\left((2a_{2,1})^{-1},...
,(2a_{2,k})^{-1}\right)^t\right)V^t. 
\end{split}
\ee
\section{Variance of \texorpdfstring{$\sigma_1$}{s1} and \texorpdfstring{$\sigma_2$}{s2}}\label{s47}
Higher moments of $\sigma_1$ and $\sigma_2$ with respect to $q$ 
can be evaluated directly as higher moments of $\sigma_1$ 
and $\sigma_2$ with respect to $\tilde{q}$.
That is, for all $j \in \{2,3,...,\}$, we have
\bb
\E_{q}((\sigma_1 - \mu_{\sigma_1})^j) = 
\E_{\tilde{q}}((\sigma_1 - \mu_{\sigma_1})^j)
\ee
and
\bb
\E_{q}((\sigma_2 - \mu_{\sigma_2})^j) = 
\E_{\tilde{q}}((\sigma_2 - \mu_{\sigma_2})^j).
\ee
In particular, for $j=2$, we obtain
\bb
\var_{q}(\sigma_1) = \var_{\tilde{q}}(\sigma_1)
\ee
and
\bb
\var_{q}(\sigma_2) = \var_{\tilde{q}}(\sigma_2).
\ee

\begin{algorithm}[H]\label{700}
\DontPrintSemicolon
\SetAlgoLined
Compute SVD of matrix $X$\;
Compute $w$ (see (\ref{110_0}))\;
Compute $V^t\1$ (see (\ref{350_0}))\;
Construct evaluator for density $\tilde{q}$ of (\ref{1270_0})\;
Evaluate first and second moments with respect to $\tilde{q}$:
$\E_{\tilde{q}}(\sigma_1),
\E_{\tilde{q}}(\sigma_2), 
\E_{\tilde{q}}(\frac{a_{1,i}}{2a_{2,i}})$\;
Compute $\E(\beta)$ via formula (\ref{570})\;
 \caption{\em Evaluation of posterior expectations of normal-normal models}
\end{algorithm}

\section{Numerical Experiments}\label{s50}
Algorithm \ref{700} was implemented in Fortran. 
We used the GFortran compiler on a 
2.6 GHz 6-Core Intel Core i7 MacBook Pro. All examples were run 
in double precision arithmetic. The matrix $X$ and 
vector $y$
were randomly generated as follows. Each entry of $X$ was generated
with an independent Gaussian with mean $0$ and variance $1$. 
The vector $y$ was created by first randomly generating a vector 
$\beta \in \R^k$, each entry of which is an independent Gaussian 
with mean $0$ and variance $1$. The vector $y$ was set to the 
value of $X\beta + \epsilon$ where $\epsilon \in \R^n$ contains 
standard normal iid entries. We generated $y$ this way in order 
to ensure that the $\E(\beta_i)$ were not all small in magnitude. 
We set $\gamma$ of (\ref{10}) to 8.

In Table \ref{7801} and Figure \ref{2270}, 
we compare the performance of Algorithm \ref{700} to two 
alternative schemes for computing posterior expectations --- 
one in which we
analytically marginalize via equation (\ref{310_0}) and 
then integrate the $2$-dimensional density via MCMC using Stan.
In the other, we use Stan's MCMC sampling
on the full $k+2$ dimensional posterior. When using MCMC with Stan, we took 10,000 posterior draws. 
In Table \ref{7801} and Figure \ref{2270} 
we denote Algorithm \ref{700} by ``SVD-Trap''. The algorithm that 
uses Stan on the marginal $2$-dimensional density is labeled
``SVD-MCMC'', and ``MCMC'' corresponds to the algorithm that uses 
only MCMC sampling in Stan.  

In the appendix, we include Stan code to sample from 
the marginal density $\tilde{q}$ of (\ref{1270_0}). 

\begin{remark}
In the numerical integration stage of algorithm \ref{700}, we use the
trapezoid rule with $200$ nodes in each direction. 
Because the integrand is smooth and vanishes near the boundary, 
convergence of the integral is super-algebraic when using the
trapezoid rule (see \cite{stoer}). A rectangular grid with $200$ 
points in each direction is satisfactory for obtaining approximately
double precision accuracy. In problems with large numbers of  
non-normal-normal parameters, MCMC algorithms such as Hamiltonian Monte 
Carlo or other methods can be used.
\end{remark}

In Tables \ref{7800} and \ref{7801}, $n$ and $k$ represent the size of 
the $n \times k$ random matrix $X$. 

The column labeled ``max error'' provides the maximum absolute
error of the expectations of $\sigma_1$, $\sigma_2$, and $\beta_i$ for $i \in 
\{1, 2, \dots, n\}$. The true solution was evaluated using trapezoid rule
with $500$ nodes in each direction in extended precision. 

In Table \ref{7800}, 
``Precompute time (s)'' denotes the time in seconds of all computations 
until numerical integration. These times are dominated by 
the cost of SVD (\ref{60}). 
The total time of the numerical integration in addition to 
the matrix-vector product (\ref{570}) is given in 
``integrate time (s).''
The final column of Table \ref{7800}, ``total time (s)'', provides
the total time of precomputation and integration. 

%
\begin{table}[!ht]

  \centering

\resizebox{\columnwidth}{!}{%

\begin{tabular}{lllccc}
   $n$ & $k$  & max error & precompute time (s) 
& integrate time (s) & total (s)\\
  \hline
  $50$ & $5$ & $0.22 \times 10^{-13}$ & $0.01$ & $0.01$ & 0.02 \\
  $100$ & $10$ & $0.26 \times 10^{-13}$ & $0.02$ & $0.01$ & 0.03\\
  $500$ & $20$ & $0.30 \times 10^{-13}$ & $0.04$ & $0.01$ & 0.05\\
  $1000$ & $50$ & $0.34 \times 10^{-13}$ & $0.09$ & $0.03$ & 0.12\\
  $5000$ & $100$ & $0.37 \times 10^{-13}$ & $0.29$ & $0.05$ & 0.34\\
  $10000$ & $500$ & $0.26 \times 10^{-13}$ & $14$ & $0.3$ & 14.2\\
  $10000$ & $1000$ & $0.39 \times 10^{-13}$ & $54$ & $0.6$ & 54.5
\end{tabular}

}
  \caption{\em Scaling of computation times for evaluation of 
expectations of $q$ (see (\ref{10})) using Algorithm \ref{700}}
    \label{7800}
\end{table}

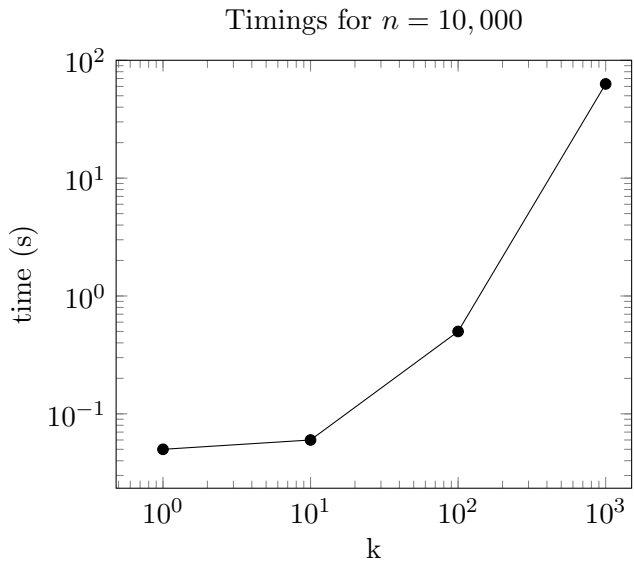
\begin{figure}[!ht]
\centering
\begin{tikzpicture}[scale=1.0]
\begin{axis}[
    xmode=log,
    ymode=log,
    xlabel={k},
    ylabel={time (s)},
    xmax=1500,
    ymax=100,
    title={Timings for $n=10,000$}
]
\addplot[mark=*]
    coordinates 
    {
 (   1.0,   0.05)
 (   10.0,   0.06)
 (   100.0,   0.5)
 (   1000.0,   63.0)
    };
\end{axis}
\end{tikzpicture}
\caption{\em Scaling of computation times for evaluation of posterior 
expectations of $q$ (see (\ref{10})) using Algorithm \ref{700} 
as a function of $k$ with $n=10,000$.}
\label{2260}
\end{figure}

\begin{table}[!h]

  \centering

\resizebox{\columnwidth}{!}{%

\begin{tabular}{ll|lc|lc|lc|}
\multicolumn{2}{c}{} 
& \multicolumn{2}{|c|}{SVD-Trap} 
& \multicolumn{2}{c|}{SVD-MCMC}
& \multicolumn{2}{c|}{MCMC}\\
\hline
   $n$ & $k$  & time (s)  & error & time (s) & error & time (s) & error\\
  \hline
 $100$ & $100$  & $0.16$ & $0.9 \times 10^{-14}$ & $11$ & $0.4 \times 10^{-4}$ & $16$ & $0.1 \times 10^{-1}$ \\
 $200$ & $100$  & $0.16$ & $0.9 \times 10^{-14}$ & $11$ & $0.3 \times 10^{-2}$ & $24$ & $0.8 \times 10^{-2}$ \\
 $500$ & $100$  & $0.23$ & $0.9 \times 10^{-13}$ & $12$ & $0.2 \times 10^{-2}$ & $40$ & $0.8 \times 10^{-2}$ \\
 $1000$ & $100$ & $0.25$ & $0.2 \times 10^{-13}$  & $12$ & $0.6 \times 10^{-3}$ & $88$ & $0.7 \times 10^{-2}$ \\
 $5000$ & $100$ & $0.30$ & $0.4 \times 10^{-13}$ & $14$ & $0.2 \times 10^{-3}$ & $617$ & $0.3 \times 10^{-2}$ \\
 $10000$ & $100$ & $0.65$ & $0.2 \times 10^{-13}$ & $13$ & $0.4 \times 10^{-3}$ & $2552$ & $0.2 \times 10^{-2}$ 
\end{tabular}
}
  \caption{\em Scaling of computation times for evaluation of 
expectations of $q$ (see (\ref{10})) using three different
algorithms: i) SVD-Trap: Algorithm \ref{700}
of this paper, ii) SVD-MCMC: marginalization with MCMC integration of 
$\tilde{q}$ using Stan, and iii) MCMC: full MCMC integration of $q$ 
using Stan.}
    \label{7801}
\end{table}

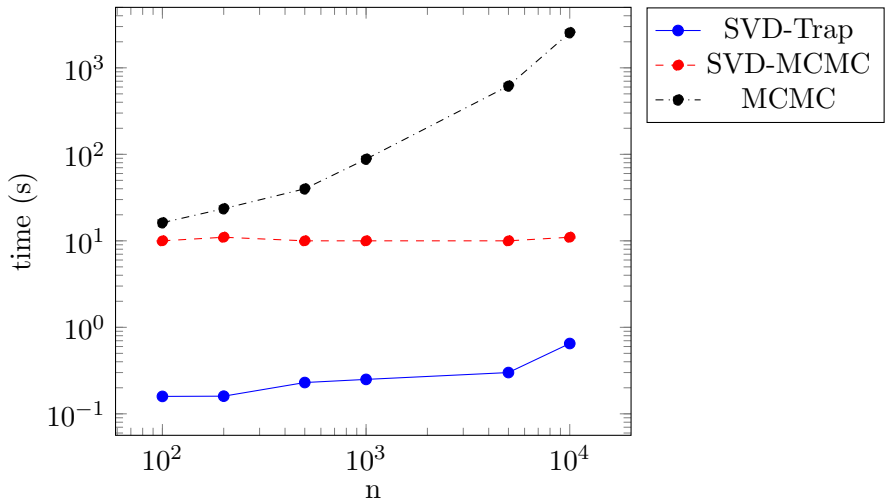
\begin{figure}[h!]
\centering
\begin{tikzpicture}[scale=1.0]
\begin{axis}[
    xmode=log,
    ymode=log,
    xlabel={n},
    ylabel={time (s)},
    xmax=20000,
    ymax=5000,
    legend pos=outer north east,
    title={Timings for $k=100$}
]
\addplot[blue, mark=*]
    coordinates 
    {
 (   100.0,   0.159)
 (   200.0,   0.16)
 (   500.0,   0.23)
 (   1000.0,   0.25)
 (   5000.0,   0.30)
 (   10000.0,   0.65)
    };
\addplot[red, mark=*,dashed]
    coordinates 
    {
 (   100.0,   10)
 (   200.0,   11)
 (   500.0,   10)
 (   1000.0,  10)
 (   5000.0,  10)
 (   10000.0, 11)
    };
\addplot[mark=*,dashdotted]
    coordinates 
    {
 (   100.0,   16.20)
 (   200.0,   23.54)
 (   500.0,   39.94)
 (   1000.0,   88.12)
 (   5000.0,   617.62)
 (   10000.0,   2552.4)
    };
\legend{SVD-Trap,SVD-MCMC,MCMC}    
\end{axis}
\end{tikzpicture}
\caption{\em Scaling of computation times for evaluation of posterior 
expectations of $q$ (see (\ref{10})) as a function of sample size $n$ 
with $k=100$. 
The three algorithms comapred are i) SVD-Trap: Algorithm \ref{700}
of this paper, ii) SVD-MCMC: marginalization with MCMC integration of 
$\tilde{q}$ using Stan, and iii) MCMC: full MCMC integration of $q$ 
using Stan.}
\label{2270}
\end{figure}

\section{Generalizations and Conclusions}\label{s55}
In this paper, we present a numerical scheme for the evaluation 
of the expectations of a particular class of distributions that
appear in Bayesian statistics; posterior dsitributions of 
linear regression problems with normal-normal parameters. 

The scheme presented generalizes naturally to 
several classes of distributions that appear frequently 
in Bayesian statistics. We list several examples of posteriors
whose expectations can be evaluated using this method.\\

\noindent 1. The choice of priors for $\sigma_1$, and $\sigma_2$ in this 
document were log normal and half-normal. This choice did not 
substantially impact the algorithm and can be generalized. 
Adaptive Gaussian quadrature can be used for the numerical integration 
step of the algorithm for a more general choice of
prior on $\sigma_1$ and $\sigma_2$. \\

\noindent 2. Multilevel regression problems with more than 
two levels. \\

\noindent 3. Regression problems with multiple groups such as 
the two-group model with posterior
\bb\label{1210}
\hspace{-3em}
\exp\left( -\frac{1}{2\sigma_1^2} 
\|X\beta - y\|^2 - \frac{1}{2\sigma_2^2}\sum_{i=1}^{k_1} (\mu_1 - \beta_i)^2 
- \frac{1}{2\sigma_3^2}\sum_{i=k_1 + 1}^{k_1 + k_2} \beta_i^2 
\right)
\ee
where $X$ is a $n \times k$ matrix, $y \in \R^n$, and 
$k_1$ and $k_2$ are non-negative integers satisfying $k_1 + k_2 = k$.\\

\noindent 3. Regression problems with correlated priors on $\beta$:
\bb\label{1215}
\hspace{-3em}
\exp\left( -\frac{1}{2\sigma_2^2} 
\|X_1\beta - y\|^2 - \frac{1}{2\sigma_1^2}\|X_2 \beta\|
\right)
\ee

For regression problems with large numbers of non-normal-normal parameters,
marginal expectations can be 
computed using, for example, MCMC in Stan. For such problems, 
the algorithm of this paper would convert an MCMC evaluation 
from $k + m$ dimensions to $m$ dimensions, where $k$ is the number
of normal-normal parameters.

\section{Acknowledgements}
The authors are grateful to Ben Bales and Mitzi Morris for 
useful discussions. 

\appendix
\section{Code}
The following Stan code allows for sampling from the distribution
corresponding to the probability density function 
proportional to (\ref{10}). 
\begin{verbatim}
data {
  int n;
  int k;
  vector[n] y;
  matrix[n,k] X;
}
parameters {
  real<lower=0> sigma1;
  real<lower=0> sigma2;
  vector<offset=0, multiplier=sigma1>[k] beta;
}
model {
  y ~ normal(X*beta, sigma2);
  beta ~ normal(0, sigma1);
  sigma1 ~ lognormal(0, 0.25);
  sigma2 ~ normal(0, 1);
}
\end{verbatim}
The following Stan program samples from the marginal density
$\tilde{q}$ (see (\ref{1270_0})). The data input \verb+yty+
corresponds to $y^ty$ of (\ref{360_0}), \verb+lam+ is the vector
of singular values of $X$, and \verb+w+ is the vector $w$ in 
equation (\ref{110_0}). We include R code for computing \verb+yty+, \verb+lam+, and \verb+w+ after the following Stan code.
\begin{verbatim}
functions {
  real q_tilde_lpdf(real sig1, real sig2, vector w, vector lam, real yty,
                    int k, int n) { 
    vector[min(n,k)] a2 = lam^2/(sig2^2) + 1/(sig1^2);
    real sol = sum(w^2 ./a2)/2/sig2^4 - sum(log(a2))/2 -yty/(2*sig2^2); 
    sol += -min(n,k)*log(sig1) - n*log(sig2);
    return sol;
  } 
}
data {
  int n;
  int k;
  vector[min(n,k)] w;
  vector[min(n,k)] lam;
  real yty;
  matrix[min(n,k),k] V;
}
parameters {
  real<lower=0> sigma1;
  real<lower=0> sigma2;
}
model {
  sigma1 ~ q_tilde(sigma2, w, lam, yty, k, n);
  sigma1 ~ lognormal(0, 0.25);
  sigma2 ~ normal(0, 1);
}
generated quantities {
 vector[k] beta;
 {
   vector[min(n,k)] zvar = 1 ./(2*(lam^2 ./(2*sigma2^2) + 1/(2*sigma1^2)));
   vector[min(n,k)] zmu = w./sigma2^2 .* zvar;
   vector[min(n,k)] z = to_vector(normal_rng(zmu, sqrt(zvar)));
   beta = V * z;
 }
}
\end{verbatim}
The following is a sample of code from R that can be used for the precomputation stage of Algorithm \ref{700}.
\begin{verbatim}
udv <- svd(X)
V <- udv$v
lam <- as.vector(udv$d)
w <- t(V) %*% t(X) %*% y
w <- as.vector(w)
yty <- t(y) %*% y
yty <- yty[1]
\end{verbatim}

\section{Proofs}
In this appendix, we include proofs of the formulas provided
in this paper. For increased readability, this appendix is 
self-contained.
\subsection{Mathematical Preliminaries and Notation}\label{s10}
In this section, we introduce notation and elementary mathematical 
identities that will be used throughout the remainder of this section. 

We define $C \in \R$ by the equation
\begin{equation}\label{20}
C = \int_{\sigma_1\in \R^{+}} \int_{\sigma_2\in\R^+} \int_{\beta \in \R^k}
q(\sigma_1, \sigma_2, \beta)
d\beta d\sigma_2 d\sigma_1,
\end{equation}
and define $\E(\sigma_1)$, $\E(\sigma_2)$, and $\E(\beta_i)$ by the formulas
\bb\label{30}
\E(\sigma_1) = \frac{1}{C} \int_{\sigma_1\in \R^{+}} \int_{\sigma_2\in\R^+} 
\int_{\beta \in \R^k}
\sigma_1 q(\sigma_1, \sigma_2, \beta)
d\beta d\sigma_2 d\sigma_1,
\ee
\bb\label{40}
\E(\sigma_2) = 
\frac{1}{C}
\int_{\sigma_1\in \R^{+}} \int_{\sigma_2\in\R^+} \int_{\beta \in \R^k}
\sigma_2 q(\sigma_1, \sigma_2, \beta)
d\beta d\sigma_2 d\sigma_1,
\ee
and
\bb\label{50}
\E(\beta_i) = 
\frac{1}{C}
\int_{\sigma_1\in \R^{+}} \int_{\sigma_2\in\R^+} \int_{\beta \in \R^k}
\beta_i q(\sigma_1, \sigma_2, \beta)
d\beta d\sigma_2 d\sigma_1
\ee
for $i \in \{1, 2, \dots, k\}$. 

We provide algorithms for the evaluation of (\ref{20}), 
(\ref{30}), (\ref{40}), and (\ref{50}). 

We will be denoting by $\1$ the vector of ones 
\bb
\1 = (1,1,\dots,1)^t.
\ee

We denote the $i^{th}$ component of a vector $v$ by $v_i$.

The following two well-known identities give the normalizing constant and 
expectation of a Gaussian distribution.
\begin{lemma}\label{250}
For all $\sigma_1, \sigma_2 > 0$ we have
\bb
\sqrt{2\pi}\sigma = \int_{\R} e^{\frac{-(\beta - \mu)^2}{2\sigma^2}}d\beta
\ee
\end{lemma}
\begin{lemma}\label{270}
For all $\mu$ in $\R$ and $\sigma > 0$, we have
\bb
\mu\sqrt{2\pi}\sigma = \int_{\R} \beta e^{\frac{-(\beta - \mu)^2}{2\sigma^2}}d\beta
\ee
\end{lemma}

\subsection{Analytic Integration of \texorpdfstring{$\beta$}{b}}
We denote the singular value decomposition of $X$ by
\begin{equation}\label{60}
X = UDV^t
\end{equation}
where $U$ is an orthogonal $n \times k$ matrix, $V$ is an orthogonal
$k \times k$ matrix, and $D$ is a $k \times k$ diagonal matrix. 
We define $z \in \R^k$ by the formula
\bb\label{70}
z = V^t\beta.
\ee
The following lemma, which will be used in the proof of Lemma \ref{170}, 
gives an expression for the second to last term of the exponent in 
(\ref{10}) after a change of variables.
\begin{lemma}\label{180}
For all $\beta \in \R^k$, and $y \in \R^n$,
\bb\label{90}
-\frac{1}{2\sigma_2^2}\| X\beta - y\|^2 = 
-\frac{y^ty}{2\sigma_2^2} 
+ \sum_{i=1}^k -\frac{\lambda_i^2}{2\sigma_2^2} z_i^2 
+ \frac{w_i}{\sigma_2^2}z_i
\ee
where 
\bb\label{110}
w = V^tX^ty,
\ee
$z$ is defined in (\ref{70}), and $\lambda_i$ is the $i^{th}$ entry on the 
diagonal of $D$ (see (\ref{60})). 
\end{lemma}
\begin{proof}
Clearly,

\bb\label{80}
\|X\beta - y\|^2 = \beta^tX^tX\beta - 2y^tX\beta + y^ty.
\ee

Substituting (\ref{60}) and (\ref{70}) into (\ref{80}), we obtain

\bb\label{160}
\begin{split}
\|X\beta - y\|^2 
&= \beta^t(UDV^t)^t(UDV^t)\beta - 2y^tXVV^t\beta + y^ty\\
&= (\beta^tV)D^2(V^t\beta) - 2y^t (V^t X^t)^t z + y^ty.
\end{split}
\ee
where $z$ is defined in (\ref{70}). Substituting (\ref{110}) and 
(\ref{70}) into (\ref{160}), we have
\bb\label{200}
\| X\beta - y\|^2 = z^tD^2z - 2w^tz + y^ty
\ee
Equation (\ref{90}) follows 
immediately from (\ref{200}). 
\end{proof}

The following lemma provides an equation for the last term 
of the exponent in (\ref{10}). The identity will be used in Lemma 
\ref{170}.

\begin{lemma}\label{190}
For all $\sigma_1>0$, 
\bb\label{120}
-\frac{\|\beta \|^2}{2\sigma_1^2} 
= \sum_{i=1}^k -\frac{z_i^2}{2\sigma_1^2} 
\ee
where $\beta \in \R^k$, $z$ is defined in (\ref{70}), and $V$ is defined in 
(\ref{60}). 
\end{lemma}
\begin{proof}
Clearly, 
\bb\label{140}
\begin{split}
\frac{\|\beta \|^2}{2\sigma_1^2} 
&= \frac{1}{2\sigma_1^2}(Vz)^t(Vz) = \frac{z^tz}{2\sigma_1^2}
\end{split}
\ee
where $V$ is defined in (\ref{60}). 
Equation (\ref{120}) follows immediately from (\ref{140}). 
\end{proof}

The following formula combines Lemma \ref{180} and Lemma \ref{190}
to convert the final two terms of (\ref{10}) into 
a Gaussian in $k$ dimensions. 
\begin{lemma}\label{170}
\bb\label{330}
\hspace{-3em}
-\frac{\| X\beta - y\|^2}{2\sigma_2^2}
- \frac{\|\beta \|^2 }{2\sigma_1^2} 
= a_0 
+ \sum_{i=1}^k 
a_{2,i}(z_i - \frac{a_{1,i}}{2a_{2,i}})^2 + \frac{a_{1,i}^2}{4a_{2,i}}
\ee
where
\bb\label{340}
a_{2,i} = \frac{\lambda_i^2}{2\sigma_2^2} + \frac{1}{2\sigma_1^2},
\ee
\bb\label{350}
a_{1,i} = \frac{w_i}{\sigma_2^2}
\ee
and
\bb\label{360}
a_{0} = -\frac{y^ty}{2\sigma_2^2}
\ee
where
$z$ is defined in (\ref{70}), $w$ is defined in (\ref{110}) and $V$ is 
defined in (\ref{60}). 
\end{lemma}

\begin{proof}
By combining Lemma \ref{180} and Lemma \ref{190}, we have
\bb\label{210}
-\frac{1}{2\sigma_2^2}\| X\beta - y\|^2 - \frac{1}{2\sigma_1^2} 
\|\beta \|^2 = 
a_0 
+\sum_{i=1}^k 
\left( 
a_{1,i}z_i - a_{2,i}z_i^2
\right).
\ee
We obtain equation (\ref{330}) by completing the square in equation
(\ref{210}). 
\end{proof}

The following theorem is the principal analytical apparatus 
of this note. It provides a formula for the $k$-dimensional 
integrals that appear in (\ref{20}), (\ref{30}), and (\ref{40}).
\begin{theorem}\label{290}
For all $\sigma_1, \sigma_2>0$
\bb\label{310}
\int_{\R^k} q(\sigma_1, \sigma_2, \beta) d\beta = 
\tilde{q}(\sigma_1, \sigma_2)
\ee
where $\tilde{q}(\sigma_1, \sigma_2)$ is defined by the formula
\bb\label{1270}
\tilde{q}(\sigma_1, \sigma_2) = 
\sigma_1^{-(k+1)} 
\sigma_2^{-n} 
\exp\left( 
-\log^2(\sigma_1) - \frac{\sigma_2^2}{2} 
+a_0 + \sum_{i=1}^k \frac{a_{1,i}^2}{4a_{2,i}}
\right)
\sqrt{2\pi}^k \prod_{i=1}^k\frac{1}{\sqrt{2a_{2,i}}}
\ee
where $a_{2,i}$ is defined in (\ref{340}), $a_{1,i}$ is defined in 
(\ref{350}) and $a_0$ is defined in (\ref{360}). 
\end{theorem}
\begin{proof}
Using (\ref{10}), clearly
\bb\label{410}
\hspace{-5em}
\int_{\R^k} q(\sigma_1, \sigma_2, \beta) d\beta = 
\sigma_1^{-(k+1)} 
\int_{\R^k} 
\exp\left( 
-\log^2(\sigma_1) - \frac{\sigma_2^2}{2} 
-\frac{1}{2\sigma_2^2}\|X\beta - y\|^2 - \frac{1}{2\sigma_1^2}\|\beta \|^2 
\right)
d\beta
\ee
Performing the change of variables (\ref{70}) and 
substituting (\ref{330}) into (\ref{410}), we have
\bb\label{670}
\hspace{-3em}
\begin{split}
\int_{\R^k} & q(\sigma_1, \sigma_2, \beta) d\beta = \\
& \exp\left(
-\log^2(\sigma_1) - \frac{\sigma^2}{2} 
+ a_0
+ \sum_{i=1}^k \frac{a_{1,i}^2}{4a_{2,i}}
\right) 
\int_{\R^k}
\exp\left(
\sum_{i=1}^k
a_{2,i}(z_i - \frac{a_{1,i}}{2a_{2,i}})^2
\right)
dz
\end{split}
\ee
Since the integrand on the right side of (\ref{670}) is a Gaussian
in $z_i$, equation (\ref{310}) follows from applying Lemma \ref{250} to 
(\ref{670}).
\end{proof}

The following theorem provides a formula for the expectation 
of $z$ (see (\ref{70})). We use this formula, in combination with
an orthogonal transformation, to obtain the expectation of $\beta$.
\begin{theorem}\label{380}
For all $\sigma_1>0$ and $\sigma_2 \in \R$,
\bb\label{390}
\int_{\R^k} (V^tx)_i q(\sigma_1, \sigma_2, \beta) d\beta = 
\frac{a_{1,i}}{2a_{2,i}}
\tilde{q}(t)
\ee
where $q$ is defined in (\ref{10}), $\tilde{q}$ is defined in (\ref{1270}),
$a_{2,i}$ is defined in (\ref{340}), $a_{1,i}$ is defined in 
(\ref{350}), $a_0$ is defined in (\ref{360}). 
\end{theorem}
\begin{proof}
Combining (\ref{670}) and (\ref{70}), we have 
\bb\label{1190}
\hspace{-3em}
\begin{split}
\int_{\R^k} & (V^t\beta)_i q(\sigma_1, \sigma_2, \beta) d\beta = \\
& 
\exp\left(
-\log^2(\sigma_1) - \frac{\sigma_2^2}{2} 
+ a_0 
+ \sum_{i=1}^k \frac{a_{1,i}^2}{4a_{2,i}}
\right)
\int_{\R^k} 
z_i\exp\left(
\sum_{i=1}^k 
a_{2,i}(z_i - \frac{a_{1,i}}{2a_{2,i}})^2
\right)
dz.
\end{split}
\ee
Applying Lemma \ref{270} to (\ref{1190}), we obtain (\ref{390}). 
\end{proof}

\end{document}